\documentclass[final, nomarks]{dmtcs-episciences}

\usepackage[utf8]{inputenc}
\usepackage{subfigure}

\usepackage[round]{natbib}

\usepackage[T1]{fontenc}

\usepackage{mathtools}
\usepackage{amsmath, amssymb, bm}
\usepackage[linesnumbered,ruled,vlined,commentsnumbered,longend]{algorithm2e}
\usepackage{xcolor}
\usepackage{comment}
\usepackage{multirow}
\usepackage{float}

\newtheorem{theorem}{Theorem}[section]
\newtheorem{corollary}[theorem]{Corollary}
\newtheorem{lemma}[theorem]{Lemma}
\newtheorem{definition}[theorem]{Definition}

\newtheorem{proposition}[theorem]{Proposition}

\SetCommentSty{mycommfont}

\SetKwInput{KwInput}{Input}     
\SetKwInput{KwOutput}{Output}   

\DeclareMathOperator*{\argmax}{argmax}

\newcommand{\then}{\Longrightarrow}

\author[I. Lamprou et al.]{Ioannis Lamprou
  \and Ioannis Sigalas
  \and Ioannis Vaxevanakis\thanks{Corresponding author, ORCID: \href{https://orcid.org/0000-0001-7345-8520}{0000-0001-7345-8520}.} \thanks{
The research work 
was supported by the Hellenic Foundation for Research and Innovation (HFRI) under the 3rd Call for HFRI PhD Fellowships (Fellowship Number: 5245).} \\
  \and Vassilis Zissimopoulos\thanks{The research work was partially supported by the Chinese Academy of
Sciences President's International Fellowship Initiative (PIFI) under
grant No. 2025PVA0150, and conducted during a visit to the Shenzhen
Institutes of Advanced Technology (SIAT).}}
  
\title[Fault-Tolerant Total and PPI Domination]{Approximations for Fault-Tolerant Total and Partial Positive Influence Domination\thanks{
Some results appearing in this paper were originally presented in ``Fault-Tolerant Total Domination via Submodular Function Approximation'' at the 17th Annual Conference on Theory and Applications of Models of Computation (TAMC), LNCS, vol.~13571, pp.~281--292, 2022.}}
\affiliation{
  Department of Informatics and Telecommunications, \\ National and Kapodistrian University of Athens, Athens, Greece}
\begin{document}
 
\publicationdata
{vol. 28:2}
{2026}
{1}
{10.46298/dmtcs.15903}
{2025-06-20; 2025-06-20; 2025-11-07}
{2025-12-11}
\maketitle
\begin{abstract}
In \textit{total domination}, given a graph $G=(V,E)$, we seek a minimum-size set of nodes $S\subseteq V$, such that every node in $V$ has at least one neighbor in $S$.
We define a \textit{fault-tolerant} version of total domination, where we require any node in $V \setminus S$ to have at least $m$ neighbors in $S$. 
Let $\Delta$ denote the maximum degree in $G$.
We prove a first $1 + \ln(\Delta + m - 1)$ approximation for fault-tolerant total domination.
We also consider fault-tolerant variants of the weighted \textit{partial positive influence dominating set} problem, where we seek a minimum-size set of nodes $S\subseteq V$, such that every node in $V$ is either a member of $S$ or the sum of weights of its incident edges leading to nodes in $S$ is at least half of the sum of weights over all its incident edges.
We prove the first logarithmic approximations for the simple, total, and connected variants of this problem.
To prove the result for the connected case, we extend the general approximation framework for non-submodular functions from integer-valued to fractional-valued functions, which we believe is of independent interest.
\end{abstract}

\section{Introduction}

\textit{Domination} is a classic graph-theoretic notion which has historically attracted much attention \citep{haynes2013fundamentals,haynes2020topics} in terms of combinatorial bounds and algorithmic complexity.
In an undirected graph, a subset of its nodes is called a \emph{dominating set} if every node of the graph is either a member of the subset or adjacent to at least one node in the subset. 
A \emph{connected dominating set} is a subset of nodes such that its induced subgraph is connected and it is a dominating set.

A variant between (simple) domination and connected domination is \textit{total domination}.
A dominating set is called total if its induced subgraph does not include any isolated nodes.  Nodes both outside and inside the set must be adjacent to at least one node in the set.
In other words, every node must have at least one neighbor in the set.
Total domination was introduced by 
\cite{cockayne1980total} as a natural variant of domination, and has been extensively studied since.
A survey of results is provided in \citep{henning2013total} and modern applications may be found in wireless (sensor) networking \citep{jena2020Total}. 

An extension to these problems is \textit{fault-tolerant domination}.
In this setting, there is a parameter $m\in\mathbb{N}$ and the goal is to find a dominating set, such that every node outside the set has at least $m$ neighbors inside.
Motivation stems from wireless sensor networks \citep{karl2007protocols}, where the failure of a few sensors should not suffice
to harm the domination property and as a result cause a system malfunction.

Another domination problem of fault-tolerant nature is the \textit{Partial Positive Influence Dominating Set} (PPIDS) problem. A subset of nodes is called a PPIDS if every node of the graph is either a member of the subset or has at least half its neighbors in the subset (positively dominated). The keyword ``partial'' implies not all nodes have to be positively dominated. The mechanism of positively dominating can be viewed as a special \textit{Linear Threshold Diffusion} model \citep{Kempe2015}. The input consists of a weighted graph $G=(V, E, w)$, where $w$ assigns weights to the edges, and an initial set of nodes $S \subseteq V$, namely the set of \textit{active} nodes. All other nodes are called \textit{inactive}. An inactive node $v\in V\setminus S$ becomes active if the sum of the weights of its active neighbors is greater than a threshold $\theta_v$, that is, it holds:
\begin{equation}\label{eq:majority}
\sum_{(v,s) \in E\,:\, s \in S}w_{(v,s)} \geq \theta_v.
\end{equation}
Weights usually represent the probability a node is influenced by a neighboring node. The goal is to find a minimum-size set of active nodes which eventually make the entire graph active. 
Choosing each node's threshold to be equal to half the weight of its incident edges is associated with the \textit{Majority Illusion} paradox \citep{Lerman2016} arising in social networks, see Figure~\ref{fig:illusion}. Essentially, an individual may believe that the behavior or opinion of a majority of their friends represents the behavior or opinion of the whole community. Thus, even a small yet appropriate starting group of like-minded individuals, corresponding to the initially active nodes, can end up influencing the opinion of the entire community.

\begin{figure}[ht]
    \begin{center}
       \includegraphics[width=100mm]{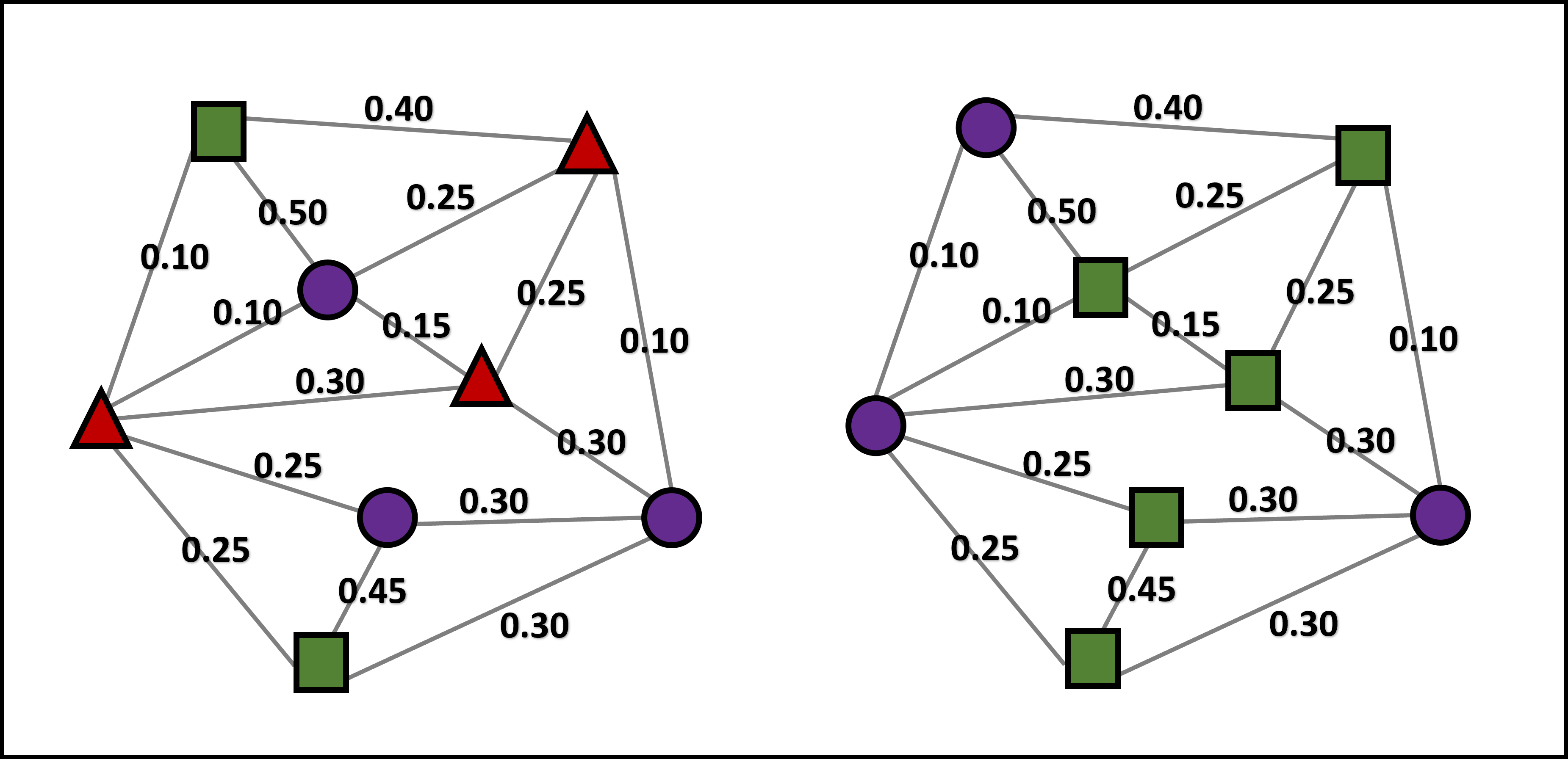}
       \caption{
        An illustration of the ``majority illusion'' paradox for an edge-weighted network. The two networks are identical, except for the three initially active (circular) nodes. Initially, all other nodes are considered inactive. In the network on the right, all initially inactive nodes satisfy inequality~\ref{eq:majority} for $\theta$ equal to 0.5, and therefore become active (square nodes). In the network on the left, however, there are inactive nodes not satisfying this threshold (triangular nodes), therefore they remain inactive.
        }
        \label{fig:illusion}
    \end{center}
\end{figure}

In this work, we consider the fault-tolerant variant of total domination and we prove a first approximation result for it. Motivated by Linear Threshold Diffusion, we consider the PPIDS problem for weighted edge graphs with rational weights. We prove a first approximation result for this problem. Moreover, we prove approximation results for the total and connected cases. To prove the result for the connected case, we develop a general approximation framework for non-submodular functions by extending the approximation technique for submodular functions found in literature.

\paragraph{Related Work.}
The complexity of dominating set problems has been studied extensively in literature and approximation algorithms have been designed.
The Minimum Dominating Set (DS), Minimum Total Dominating Set (TDS) and Minimum Connected Dominating Set (CDS) problems are all NP-hard and there is no polynomial time algorithm with approximation ratio $(1-\varepsilon)\ln|V|$, for any $\varepsilon > 0$, unless $NP\subseteq DTIME\big(|V|^{O(\log\log |V|)}\big)$ \citep{Guha1998,CHLEBIK20081264}.

Recall a set $S$ is a \textit{fault-tolerant dominating set} if every node in $V\setminus S$ has at least $m$ neighbors in $S$. In \citep{foerster2013approximating}, a greedy algorithm with a submodular function is used to approximate fault-tolerant DS. Instead of a finite sums approach, like the one we use in this paper, their analysis employs an estimation formula to achieve a $1+\ln(\Delta + m)$ guarantee.

A relevant variant, with the same approximation hardness as DS, is $k$-tuple domination, where $S$ is a $k$-tuple dominating set if every node in $V\setminus S$ has at least $k$ neighbors in $S$ and every node in $S$ has at least $k-1$ neighbors in $S$.
The problem is introduced by \cite{KLASING200475}, where they apply a reduction to Minimum $k$-Cover, a budget variant of Set Cover, to obtain a $1+\ln(\Delta+1)$ guarantee.

A greedy approach for TDS in \citep{zhu2009approximation} yields a $1.5+\ln(\Delta-0.5)$ approximation by using a potential function formed as the sum of a submodular function and a non-submodular function. Later, in subsection~\ref{sec:total-domination}, we further discuss this result in comparison to the methodology we follow in this paper.
An improved result for TDS is given by \cite{CHLEBIK20081264}, where the problem is reduced to Set Cover and an $H(\Delta)-0.5$ approximation is obtained with $H$ being the harmonic function.
Note it holds $$\ln(n)+\frac{1}{2n}+\gamma>H(n)>ln(n)+\gamma,$$ where $\gamma=0.5772156649$ is the Euler constant. 

Similarly to TDS, \cite{Guha1998} prove the first $2+H(\Delta)$ approximation for CDS. 
A better approximation for CDS is given by \cite{Ruan2004} via a greedy algorithm with a non-submodular function proved to obtain a $2+\ln(\Delta)$ approximation. 
The best approximation for CDS is given by \cite{Du2008}, where they present a $(1+\varepsilon)(1+\ln(\Delta-1))$ guarantee. 
For fault-tolerant CDS, \cite{zhang2009two} give a $2H(\Delta+m-1)$ approximation algorithm. Later, \cite{Zhou2014} improve this result to $2+\ln(\Delta+m-2)$. To achieve that, they use the same potential function as in \citep{Ruan2004}, ensuring connectivity, and they add an additional function to count the extra neighbors.

A more general problem of the fault-tolerant CDS is the $k$-connected $m$-fold dominating set problem, denoted as $(k,m)$-CDS. Given a graph $G = (V, E)$, a subset of nodes $S \subseteq V$ is a $(k,m)$-CDS if every node in $V \setminus S$ has at least $m$ neighbors in $S$, and the subgraph induced by $S$ is $k$-connected. A $(2k-1)\alpha_0$-approximation algorithm for this problem is proposed by \cite{zhang2018computing}, where $\alpha_0$ denotes the approximation ratio for the minimum $(1, m)$-CDS problem (the $(1, m)$-CDS problem is exactly the fault-tolerant CDS problem).

A variant of fault-tolerant domination is Partial Positive
Influence Dominating Set (PPIDS). A set $S$ is called a PPIDS if every
node in $V\setminus S$ has at least half of its
neighbors in $S$. In the connected case, a set $S$ is called a PPICDS if $S$ is PPIDS and connected. \cite{zhu2010new} prove a $1+\ln(\lceil\frac{3\Delta}{2}\rceil)$ approximation for PPIDS and a $2+\ln(\lceil\frac{5\Delta}{2}\rceil)$ approximation for PPICDS by using a simple greedy algorithm. Later, \cite{zhong2023unified} define generalizations of PPIDS and PPICDS, which require each node outside the dominating set to have its number of neighbors inside the dominating set equal to a certain percentage $p$ of their degree (degree percentage constraints), and prove $1+\ln(\Delta+p\cdot\Delta)$ and $2+\ln(2\Delta+p\cdot\Delta)$ approximation, respectively. In the tuple version of the percentage constraint of PPIDS (referred to as "Total PIDS" in \citep{dinh2014approximability}, although their problem is not the same as the one we consider), where the goal is to find a set $S \subseteq V$ such that every node in the graph has a number of neighbors in $S$ equal to a certain percentage $p$ of its degree, the authors provide a $(1 + \ln \Delta)$-approximation guarantee.

For many of the above discussed results, the methods used boil down to the \textit{Set Function Optimization} problem. In general, if there is a finite universe $U$ and a function $f: 2^U \rightarrow \mathbb{R}$, the goal is to find a set $S\in U$ with specific properties that maximizes or minimizes the function. The first approximation results for finding the minimum size set such that a submodular function $f$ takes the maximum value are provided in \citep{wolsey1982analysis} by using a greedy algorithm. If $S=\{s_1,s_2, \cdots, s_{|S|}\}$ the solution of the greedy algorithm, where each $s_i$ denotes the $i$-th element selected by the algorithm, the authors provide an $H(\delta_{\max})$ approximation for integer functions and $1+\ln(\delta_{\max}/\delta_{\min})$ aproximation for real functions, where:

\[ \delta_{\max}=\Delta_{s_1}f(\emptyset)=\max\limits_{x\in U}\big(\Delta_{x}f(\emptyset)\big),\]\[ \delta_{\min}=\Delta_{s_{|S|}}f(S\setminus \{s_{|S|}\})\geq\min\limits_{\substack{x\in U\setminus A, \ A\subseteq U \\ \Delta_{x}f(A)>0}}\big(\Delta_{x}f(A)\big),\]

\noindent $\Delta_{x}f(A)=f(A\cup \{x\})-f(A)$, and $H$ is the harmonic function (here it holds $\delta_{\max}=\Delta$ and $\delta_{\min}\geq 1$). Improved results provided in \citep{wan2010greedy} and \citep{chen2022general}, establish $1+\ln(f(U)/opt)$ and $1+\ln\big(f(U)/(\delta_{\min}\cdot opt)\big)$ approximation respectively, where $opt$ stands for the size of an optimal set.

Further approximation results exist for different function properties such as $\gamma$-weakly submodular \citep{SHI2021126442} and $\varepsilon$-approximately submodular \citep{qian2019maximizing} functions.

In \citep{SHI2021126442}, the authors define the $\gamma$-weakly submodular function and provide several approximation results. A function is called $\gamma$-weakly submodular ($\gamma\geq 1$) if it satisfies the property $$f(B\cup \{x\})-f(B)\leq \gamma[f(A\cup \{x\})-f(A)],$$ for all $A\subseteq B\subseteq U$ and $x\in U\setminus B$. This property is also denoted as the diminishing returns ratio (\textit{DR ratio}). For the problem of finding the minimum size set that maximizes an Integer-Valued $\gamma$-weakly submodular function, \cite{SHI2021126442} provide the first $1/\gamma+\ln(\delta_{\min})$ approximation result, where $\delta_{\min}$ is the maximum function value over all singletons. Also, they generalized the result for Fraction-Valued functions.

A function is $\varepsilon$-approximately submodular ($\varepsilon\geq 0$) if and only if $$f(B\cup \{x\})-f(B)\leq f(A\cup \{x\})-f(A)+\varepsilon,$$ for all $A\subseteq B\subseteq U$ and $x\in U\setminus B$. This property is also denoted as the diminishing returns gap (\textit{DR gap}). \cite{qian2019maximizing} provide a weak approximation result ($f(S)\geq (1-1/e)(opt-k\varepsilon)$) for the problem of finding a set with at most $k$ elements that maximizes an $\varepsilon$-approximately submodular function. Also, \cite{SHI2021126442} provide the first $1+\varepsilon+\ln(\delta_{\min})$ approximation result for the problem of finding the minimum size set that maximizes an Integer-Valued $\varepsilon$-approximately submodular function.

\paragraph{Our Results.}
In Section~\ref{sec:ggs}, we develop a general framework using $\varepsilon$-approximately submodular definition and providing approximation results for a variant of Set Function Optimization. Given a universe of elements $U$ and an $\varepsilon$-approximately submodular set function $f:2^U\rightarrow \mathbb{R}$, the goal is to find a minimum size set $S^*\subseteq U$ such that $f(S^*)$ is maximum. By using a greedy algorithm we obtain a $\big(1+\varepsilon/\delta_{\min}+\ln(\delta_{\max}/\delta_{\min})\big)+O(1)$-approximation result. This result extends the approximation guarantee of \citep{SHI2021126442} from integer-valued $\varepsilon$-approximately submodular functions to the case of fraction-valued $\varepsilon$-approximately submodular functions. Additionally, we introduce a more general notion of $\varepsilon$-approximate submodularity and prove that the same approximation guarantee holds.

In Section~\ref{sec:total}, we define the fault-tolerant total domination problem. For this problem, we define a submodular function deciding if a set is a total dominating set. In other words, for every total dominating set the function takes the maximum value. By using a greedy algorithm we obtain a first $1 + \ln(\Delta + m - 1)$ approximation result. 

In Section~\ref{sec:WPPI}, we generalize partial positive influence domination problems, the simple, total, and connected versions, by allowing non-integer, that is, rational weights (WPPIDS, WPPITDS, and WPPCDS respectively). We obtain approximation results for these problems by using a greedy algorithm. In subsection~\ref{sec:WPPIDS}, we define a submodular function that decides if a set is WPPIDS and we obtain a first $\big(1+\ln(\frac{3}{2}\cdot L \cdot W)\big)$ approximation result, where $W$ is equal to the maximum sum of weights of a node's incident edges, that is, an extension of maximum node degree $\Delta$, and $L$ is a special factor depending on the weights. In subsection~\ref{sec:WPPITDS}, we define a submodular function that decides if a set is WPPITDS and we obtain a first $\big(1+\ln(\frac{3}{2}\cdot L \cdot W+\Delta)\big)$ approximation result. In subsection~\ref{sec:WPPICDS}, we define a non-submodular function that decides if a set is WPPICDS. The function has a special property called \textit{conditional submodularity gap}, which is described in subsection~\ref{sec:con-sub-gap}. By applying the framework in subsection~\ref{sec:Submodularity gap Maximization}, we obtain a first $\big(2+\ln(\frac{3}{2}\cdot L\cdot W + \Delta )\big)$ approximation result.  All results are a generalization of \citep{zhu2010new} (for unit weights, it holds $W=\Delta$ and $L=1$) and can be extended to degree percentage constraints problems \citep{zhong2023unified} (by simple modification of the potential function).

\section{Preliminaries}\label{sec:prel}

Let $G=(V, E, w)$ a weighted graph, where $w: E\rightarrow\mathbb{R}^+$.
The \textit{open neighborhood} of node $v$ is the set of all its neighbors and is denoted by $N(v)$.
Let the maximum degree of the graph be denoted by $\Delta = \max\limits_{v \in V} |N(v)|$.
The open neighborhood of $v$ in the subset of nodes $C\subseteq V$ is denoted by $N_C(v)=N(v)\cap C$. 
For a singleton set $\{x\}$, we simplify the notation from $N_{\{x\}}(v)$ to $N_x(v)$. Also, for a weighted graph, we denote $W_A(v)=\displaystyle\sum_{\substack{i\in N_A(v)}} w_{(v, i)}$, $W(v) = W_V(v)$ and $W=\max\limits_{v\in V}W(v)$.

Let $U$ denote a universe of elements.
The set of all subsets of $U$ is denoted by $2^U$.
A function $f:2^U\rightarrow\mathbb{R}$ is called \textit{non-decreasing} if for any $A\subseteq B\subseteq U$ it holds $f(A)\leq f(B)$. 
Let $\Delta_{x}f(A)=f(A\cup \{x\})-f(A)$.
A function $f:2^U\rightarrow\mathbb{R}$ is \textit{submodular} if for any $A\subseteq B\subseteq U$, $x \in U\setminus B$, it holds $\Delta_{x}f(B)\leq \Delta_{x}f(A)$. The maximum value of a function is denoted $f_{\max}=\max\limits_{{X\subseteq U}}f(X)$.

A set $S \subseteq V$ is a \textit{Dominating Set} (DS) of $G$ if every node in $V\setminus S$ has at least one neighbor in $S$. The minimum dominating set is the problem of finding a DS with minimum cardinality. A set 
$S$ is a \textit{Fault-Tolerant Dominating Set} ($m$-DS) if every node in $V\setminus S$ has at least $m$ neighbors in $S$. The minimum fault-tolerant dominating set is the problem of finding a m-DS with minimum cardinality. A set
$S$ is a \textit{Total Dominating Set} (TDS) if $S$ is a dominating set and every node in $S$ has at least one neighbor in $S$.
Equivalently, every node in $V$ has at least one neighbor in $S$. The minimum total dominating set is the problem of finding a TDS with minimum cardinality. A set
$S$ is a \textit{Fault-Tolerant Total Dominating Set} (m-TDS) if it is fault-tolerant dominating and every node in $S$ has at least one neighbor in $S$. The minimum fault-tolerant total dominating set is the problem of finding a m-TDS with minimum cardinality.

A set $S \subseteq V$ is a \textit{Partial Positive Influence Dominating Set (PPIDS)} of $G$ if every node $v\in V\setminus S$ has at least $\lceil N(v)/2 \rceil$ neighbors in $S$. In the weighted version, a set $S \subseteq V$ is a \textit{Weighted Partial Positive Influence Dominating Set (WPPIDS)} of $G=(V,E,w)$ if for every node $v\in V\setminus S$ it holds $W_S(v)\geq  W(v)/2$. The minimum weighted partial positive influence dominating set is the problem of finding a WPPIDS with minimum cardinality. In the total version, the extra condition is the set $S$ must be total (WPPITDS). The minimum weighted partial positive influence total dominating set is the problem of finding a WPPITDS with minimum cardinality. In the connected version, the extra condition is the set $S$ must be connected (WPPICDS). The minimum weighted partial positive influence connected dominating set is the problem of finding a WPPICDS with minimum cardinality.

Let $U$ be a finite set (universe) of elements and $f:2^U\rightarrow\mathbb{R}$. In the Function Maximization Problem, we wish to find a minimum-cardinality subset $S\subseteq U$ of elements that maximizes function $f$.

\section{General Greedy Non-Submodular Approximation}\label{sec:ggs}

In this section, we present a greedy approximation algorithm for a variant of Set Function Optimization, when the function is non-decreasing and $\varepsilon$-approximately submodular.

Let $U$ be a finite set (universe) of elements.
For our purposes, a problem with universe $U$ may be defined as a function $Q: 2^U \rightarrow \{0,1\}$, such that:
\begin{itemize}
    \item[$\bullet$] $Q(S) = 1$ $\iff$ $S\subseteq U$ is a feasible solution to the problem,
    \item[$\bullet$] $Q(S) = 0$ otherwise.
\end{itemize}

Let $\mathcal{Q} = \{S \subseteq U: Q(S) = 1
\}$.
Assume there is a function $f:2^{U}\rightarrow\mathbb{R}$ with the property:
\[
\big(\forall C\subseteq U\big) \ f(C)=\max\limits_{{U'\subseteq U}}f(U') \iff C\in \mathcal{Q}.
\]

Let $S^*$ be a member of $\mathcal{Q}$ with minimum cardinality.
In Algorithm~\ref{alg:GC}, we introduce the well-known Greedy Constructor to return a solution approximating the size of $S^*$.

 \begin{algorithm}[ht]\label{alg:GC}
    \caption{Greedy Constructor}
    \SetAlgoLined
    \DontPrintSemicolon
    
        \KwInput{Universe $U$}
        \KwOutput{$S \in \mathcal{Q}$}
        $S \gets \emptyset$\; 
        \While{$\big(\exists u\in U\setminus S\big) \ \Delta_{u}f(S)>0$}{
            $x \gets \argmax\limits_{{u\in U \setminus S}}\Delta_{u}f(S)$\;
            $S \gets S\cup\{x\}$
        }
        Return $S$
        
    \end{algorithm}

To begin with the analysis of Greedy Constructor, we first define what is a  \textit{Greedy Maximum Differential Set}, which will be returned by the algorithm.

\begin{definition}[Greedy Maximum Differential Set]\label{def:gmds}
    Let $U \subseteq \mathbb{N}$ and $f:2^U\rightarrow\mathbb{R}$.
    Let $S\subseteq U$ and a total order of the elements in $S$, namely $s_1, s_2, \ldots, s_{|S|}$.
    Let $S_{i} =\{s_1, s_2,...,s_{i}\}\subseteq S$, where $S_0 = \emptyset$.
    A set $S\subseteq U$ is called a \textit{Greedy Maximum Differential Set} of $f$ if there is a total order of the elements of $S$ such that for each $s_i$ it holds $\Delta_{s_{i}}f(S_{i-1})\geq \Delta_{u}f(S_{i-1})$ for all $u\in U$.
 \end{definition} 
    
\begin{definition}
    The set of all Greedy Maximum Differential Sets of $f$ is denoted by $D_f$.
\end{definition}

Moving on, we give a definition for a special case of non-submodular functions called $\varepsilon$-approximately submodular \citep{qian2019maximizing}.

  \begin{definition}[$\varepsilon$-approximately submodular]\label{def:b-gap}
    Let function $f:2^U\rightarrow\mathbb{R}$. The function $f$ is $\varepsilon$-approxima-tely submodular if there is $\varepsilon\in [0,+\infty)$ such that for every $A\subseteq B\subseteq U$ and $x \in U\setminus B$ it holds $\Delta_{x}f(B)\leq \Delta_{x}f(A)+\varepsilon$. The minimum value among all $\varepsilon\in [0,+\infty)$ such that $f$ is $\varepsilon$-approximately submodular is called the \textit{submodularity gap}.
 \end{definition}

The next corollary and theorem are some observations about $\varepsilon$-approximately submodular functions.

\begin{corollary}\label{cor:b=0_submodular}
    A function $f$ is $0$-approximately submodular if and only if $f$ is submodular.
\end{corollary}

\begin{theorem}\label{thm:beta}
    Every function $f:2^U\rightarrow\mathbb{R}$, where $U$ is finite, is $\varepsilon$-approximately submodular.
    
\end{theorem}

\begin{proof}
    Let $D=\{d\in\mathbb{R}: d= \Delta_xf(A), \ x\in U \ and \ A\subseteq U\}$. Because $U$ is finite then the set $D$ has minimum and maximum. Let $\varepsilon=\mu-\lambda$, where $\mu=\max(D)$ and $\lambda=\min(D)$. Then, for every $A\subseteq B\subseteq U$ and $x\in U$, it holds:
    \begin{align*}
        \lambda\leq\Delta_xf(A) \ and \ \Delta_xf(B) \leq \mu = \lambda + \mu-\lambda\leq \Delta_xf(A) + \varepsilon.
    \end{align*}
\end{proof}

\subsection{Non-Submodular Maximization with Submodularity gap}\label{sec:Submodularity gap Maximization}

    Let $f$ be a non-decreasing and $\varepsilon$-approximately submodular function.
    In the following key Lemma, we bound the size of a greedy maximum differential set $S$ to be a logarithmic approximation of the size of an optimal solution achieving $f_{\max}$. Let $S\in D_f$. Since the function is not submodular, we redefine the values of $\delta_{\max}$ and $\delta_{\min}$ as follows:
    \[\delta_{\max}=\Delta_{s_1}f(\emptyset)=\max\limits_{x\in U}\big(\Delta_{x}f(\emptyset)\big), \ \ \delta_{\min}= 
    \min\limits_{\substack{i\in\{1,\cdots,|S|-1\}}}\big(\Delta_{s_i}f(S_{i-1})\big)\geq
    \min\limits_{\substack{x\in U\setminus A, \ A\subseteq U \\ \Delta_{x}f(A)>0}}\big(\Delta_{x}f(A)\big).
    \]
     
    \begin{lemma}\label{lem:gap-greedy-approximation}
        Let $U$ be a finite set and $f: 2^U\rightarrow\mathbb{R}$ non-decreasing and $\varepsilon$-approximately submodular function.
        Let $S\subseteq U$ with the properties:
        \[
            i) \ S\in D_f, \ \ \ ii) \ f(S)=f_{\max}, \ \ \ iii) \ \delta_{\min}> 0.
        \]
        For every set $C\subseteq U$ with the property $f(C)=f_{\max}$, it holds
        \[
        |S|<\Bigg(1+\frac{\varepsilon}{\delta_{\min}}+\ln\bigg(\frac{\delta_{\max}}{ \delta_{\min}}\bigg)\Bigg)\cdot |C|+1.
        \]
    \end{lemma}
    \begin{proof}
        Let $S=\{s_1, s_2, ... , s_{|S|}\} \in D_f$ such that $f(S) = f_{\max}$ and $\delta_{\min} > 0$.
        Let $C=\{c_1, c_2, ... , c_{|C|}\}\subseteq U$ with $f(C)=f_{\max}$, where $c_1, c_2, ... , c_{|C|}$ is an arbitrary order of the elements of C.
        Let $S_i=\{s_1, s_2, ... , s_i\}$, and $C_i=\{c_1, c_2, ... , c_i\}$, where $S_0=C_0=\emptyset$. 
    
        Since $f(C) = f_{\max}$ and $f$ is non-decreasing, we get $f(S_i\cup C) = f_{\max}$ for any set $S_i$.
        
        It follows,
        \begin{align*}
            f_{\max}-f(S_{i-1}) &= f(S_{i-1}\cup C) - f(S_{i-1})\\
            &= f(S_{i-1}\cup C_{|C|}) - f(S_{i-1} \cup C_0)\\
            &= f(S_{i-1}\cup C_{|C|})\\
            & \quad -f(S_{i-1}\cup C_{|C|-1})+f(S_{i-1}\cup C_{|C|-1})\\
            & \hspace{1.5in}\vdots\\
            & \quad -f(S_{i-1}\cup C_{1})+f(S_{i-1}\cup C_{1})\\
            & \quad - f(S_{i-1}\cup C_0)\\
            &= \sum_{j=1,\ldots,|C|}f(S_{i-1}\cup C_{j})-f(S_{i-1}\cup C_{j-1})\\
            &= \sum_{j=1,\ldots,|C|}\Delta_{c_j}f(S_{i-1}\cup C_{j-1})
        \end{align*}
        since by definition $C_{|C|} = C$, $C_0 = \emptyset$.
        
        Since $f$ is $\varepsilon$-approximately submodular, for any $c_j$, it follows: $$\Delta_{c_j}f(S_{i-1}\cup C_{j-1}) \le \Delta_{c_j}f(S_{i-1})+\varepsilon.$$
        Let $c_{j'} \in C$ be the element maximizing $\Delta_{c_j}f(S_{i-1})$.
        Then, 
        \begin{align*}
        \sum_{j=1,...,|C|}(\Delta_{c_j}f(S_{i-1})+\varepsilon)&\leq |C|\cdot(\Delta_{c_{j'}}f(S_{i-1})+\varepsilon)\\
        &\leq |C|\cdot(\Delta_{s_i}f(S_{i-1})+\varepsilon),
        \end{align*}
        where $\Delta_{c_{j'}}f(S_{i-1}) \le \Delta_{s_i}f(S_{i-1})$, since $S \in D_f$ (Definition~\ref{def:gmds}).
        
        Overall, we have
        \begin{align}
            f_{\max}-f(S_{i-1})&\leq|C|\cdot(\Delta_{s_i}f(S_{i-1})+\varepsilon)\nonumber\\
            \frac{f_{\max}-f(S_{i-1})}{|C|}&\leq f(S_i)-f(S_{i-1})+\varepsilon \label{eq:b-gap}\\
            -f(S_i)&\leq -f(S_{i-1})-\frac{f_{\max}-f(S_{i-1})}{|C|} +\varepsilon \nonumber\\
            f_{\max}-f(S_i)&\leq f_{\max}-f(S_{i-1})-\frac{f_{\max}-f(S_{i-1})}{|C|}+\varepsilon \nonumber
        \end{align}
        Let $a_i = f_{\max}-f(S_i)$ for any $i$. Then, by induction, it follows:
        \begin{align}
        a_i&\leq a_{i-1}-\frac{a_{i-1}}{|C|}+\varepsilon =
        a_{i-1}\bigg(1-\frac{1}{|C|}\bigg)+\varepsilon \leq
        \cdots 
        \leq a_0\bigg(1-\frac{1}{|C|}\bigg)^i +\varepsilon\cdot\sum_{k=0}^{i-1}\bigg(1-\frac{1}{|C|}\bigg)^k \nonumber\\
        &=a_0\bigg(1-\frac{1}{|C|}\bigg)^i +\varepsilon\cdot |C| \Bigg(1-\bigg(1-\frac{1}{|C|}\bigg)^i\Bigg)=
            (a_0-\varepsilon\cdot |C|)\bigg(1-\frac{1}{|C|}\bigg)^i+\varepsilon\cdot |C| \nonumber\\
         &\leq (a_0-\varepsilon\cdot |C|)\cdot e^{-\frac{i}{|C|}}+\varepsilon\cdot |C| \nonumber
        \end{align}
        \noindent since for any $x \in \mathbb{R}$ it holds $(1+x) \le e^{x}$.
        
        \begin{proposition}\label{claim:S-upper}
            Let $\delta>0$.
            For every $k\in \{0, ... ,|S|\}$, if $a_{k}\leq \delta\cdot |C|$, then $|S|\leq \frac{\delta\cdot |C|}{\delta_{\min}}+k$. Also, if $a_{k}< \delta\cdot |C|$, then $|S|< \frac{\delta\cdot |C|}{\delta_{\min}}+k$.
        \end{proposition}
        \begin{proof}
            We show the proof only for the case [$\leq$]. The case [$<$] follows similarly.
            
            We first show $\delta_{\min}\cdot (|S|-k)\leq a_k$.
            We rewrite the left part as: $$\delta_{\min}\cdot (|S|-k) = \sum\limits_{i=k+1}^{|S|} \delta_{\min}\leq \sum\limits_{i=k+1}^{|S|}\Delta_{s_i}f(S_{i-1}),$$ where the inequality follows by definition of $\delta_{\min}$.
            We proceed to show equality of the right part to $a_k$:
            \begin{align*}
             &\sum\limits_{i=k+1}^{|S|}\Delta_{s_i}f(S_{i-1})=
                \sum\limits_{i=k+1}^{|S|} f(S_{i})-f(S_{i-1})= \\
                &= f(S_{|S|})
                +\big(-f(S_{|S|-1})+f(S_{|S|-1})\big)
                +\cdots
                +\big(-f(S_{k+1})+f(S_{k+1})\big)
                - f(S_{k})\\
                &=f(S_{|S|})-f(S_{k})\\
                &=f_{\max}-f(S_{k}) =a_{k}.
            \end{align*}
            Since by assumption $a_k \le \delta\cdot |C|$, it follows $\delta_{\min}\cdot (|S|-k) \le \delta\cdot|C|$, which completes the proof.
        \end{proof}
        
        We continue the proof of the Lemma. For some $\delta > \varepsilon \geq 0$, we distinguish two cases. 
        
        \noindent If $a_0\leq \delta\cdot |C|$, then by Proposition~\ref{claim:S-upper} it holds $|S|\leq \frac{\delta\cdot |C|}{\delta_{\min}}$ (case I).
        
        \noindent If $a_0> \delta\cdot |C|$, then since $a_i$ is decreasing (because $\delta_{\min}>0$), by definition there exists $i_0$ such that: $$a_{i_0+1}< \delta\cdot |C|\leq a_{i_0}.$$
        We first upper bound $i_0$ based on the right inequality. We have:
        \begin{align}
        a_{i_0} &\le (a_0-\varepsilon\cdot |C|)\cdot e^{-\frac{i_0}{|C|}}+\varepsilon\cdot |C| \iff\nonumber \\
            \delta\cdot |C|&\leq (a_0-\varepsilon\cdot |C|)\cdot e^{-\frac{i_0}{|C|}}+\varepsilon\cdot |C| \iff\nonumber \\
            e^{\frac{i_0}{|C|}} &\leq \frac{(a_0-\varepsilon\cdot |C|)}{(\delta-\varepsilon)\cdot|C|} \iff\nonumber\\
            i_0&\leq |C|\ln\bigg(\frac{a_0-\varepsilon\cdot |C|}{(\delta-\varepsilon)\cdot|C|}\bigg).
            \label{eq:i_0}
        \end{align}
        The last inequality stands because $\delta>\varepsilon$ and so $a_0>\delta\cdot|C|>\varepsilon\cdot|C|$. To complete the proof, we now upper bound $S$ based on left inequality.
        Since $a_{i_0+1}< \delta\cdot |C|$, by Proposition~\ref{claim:S-upper}, it follows $|S| < \frac{\delta\cdot |C|}{\delta_{\min}}+i_0+1$ (case II).
        
        Note that it suffices to examine case II, since it contains case I. It follows:
        \begin{align}
            |S|&< \frac{\delta\cdot |C|}{\delta_{\min}}+i_0+1 \leq^\eqref{eq:i_0}
            \frac{\delta\cdot |C|}{\delta_{\min}}+|C|\ln\bigg(\frac{a_0-\varepsilon\cdot |C|}{(\delta-\varepsilon)\cdot|C|}\bigg) + 1 \iff\nonumber
            \\
            |S|&<
            \Bigg(\frac{\delta}{\delta_{\min}}+\ln\bigg(\frac{a_0-\varepsilon\cdot |C|}{(\delta-\varepsilon)\cdot|C|}\bigg)\Bigg)\cdot |C| + 1.  \nonumber
        \end{align}
        
        Now we find a value of $\delta$ which minimizes the upper bound of $S$. Let $g(\delta) = \Bigg(\frac{\delta}{\delta_{\min}}+\ln\bigg(\frac{a_0-\varepsilon\cdot |C|}{(\delta-\varepsilon)\cdot|C|}\bigg)\Bigg)$. The function $g$ is minimized when $\delta = \delta_{\min}+\varepsilon$ because:
        
        \begin{align*}
         g'(\delta) = \frac{1}{\delta_{\min}}-\frac{1}{\delta-\varepsilon}
         \ \ \text{and} \ \
         g'(\delta) = 0 \iff \delta = \delta_{\min}+\varepsilon.
        \end{align*}
        
        \noindent So,
        
        \noindent\resizebox{\textwidth}{!}{$|S|<
            \Bigg(\frac{\delta_{\min}+\varepsilon}{\delta_{\min}}+\ln\bigg(\frac{a_0-\varepsilon\cdot |C|}{\delta_{\min}\cdot|C|}\bigg)\Bigg)\cdot |C| + 1 \leq 
            \Bigg(1+\frac{\varepsilon}{\delta_{\min}}+\ln\bigg(\frac{\frac{a_0}{|C|}-\varepsilon}{\delta_{\min}}\bigg)\Bigg)\cdot |C| + 1$}
            
            $$\iff^{(\ref{eq:b-gap})} \ |S|<
            \Bigg(1+\frac{\varepsilon}{\delta_{\min}}+\ln\bigg(\frac{\delta_{\max}}{\delta_{\min}}\bigg)\Bigg)\cdot |C| + 1.$$
    \end{proof}

    In the next theorem, we apply Lemma~\ref{lem:con-gap-greedy-approximation} to provide an approximation result for the function maximization problem, where the function $f$ is non-decreasing, and is $\varepsilon$-approximately submodular. So, for any problem defined with a potential function which is $\varepsilon$-approximately submodular, we can prove a logarithmic approximation result.
    
    \begin{theorem}\label{thm:beta-greedy}
        Algorithm~\ref{alg:GC} returns a \resizebox{0.8\width}{!}{$\bigg(1+\frac{\varepsilon}{\delta_{\min}}+\ln\Big(\frac{\delta_{\max}}{\delta_{\min}}\Big)\bigg)$}$+O(1)$-approximation to the minimum size set in $\mathcal{Q}$, if and only if $f$ is non-decreasing and $\varepsilon$-aproximately submodular function and $f(S) = f_{\max}$, where $S$ is the solution returned by the Algorithm.
    \end{theorem}

    \begin{proof}
        By definition of Algorithm \ref{alg:GC} and function $f$, it holds
        $S\in D_f, \ f(S)=f_{\max}$, and $\delta_{\min} > 0$.
        If $C$ is a set of minimum cardinality within $\mathcal{Q}$, by Lemma~\ref{lem:gap-greedy-approximation} it holds
        \[|S|< \Bigg(1+\frac{\varepsilon}{\delta_{\min}}+\ln\bigg(\frac{\delta_{\max}}{\delta_{\min}}\bigg)\Bigg)\cdot |C| + 1.
        \]
    \end{proof}
    
    By using the next corollary, we can simplify some approximation results.
    \begin{corollary}\label{cor:simplify-big-o}
        If $\varepsilon$ is multiple of $\delta_{\min}$ then the approximation result of Theorem~\ref{thm:beta-greedy} is simplified to
        
        \noindent$\Big(1+\frac{\varepsilon}{\delta_{\min}}+\ln\big(\frac{\delta_{\max}}{\delta_{\min}}\big)\Big)$. 
    \end{corollary}
    \begin{proof}
        We recall the inequality $|S|< \frac{\delta\cdot |C|}{\delta_{\min}}+i_0+1$.
        
        In Lemma~\ref{lem:gap-greedy-approximation}, we set $\delta=\delta_{\min}+\varepsilon$. So, if $\varepsilon=k\cdot\delta_{\min}$ with $k\in\mathbb{N}$ then the factor $\frac{\delta}{\delta_{\min}}\in\mathbb{N}$ and the result holds.
    \end{proof}

    The above theorem is a generalization of the classical theorem for submodular functions to non-submodular functions. In the case of submodular functions, the same result is obtained. 

    \begin{corollary}\label{cor:b-to-submodular}
        If $f$ is submodular function then the approximation result of Theorem~\ref{thm:beta-greedy} is simplified to $\Big(1+\ln\big(\frac{\delta_{\max}}{\delta_{\min}}\big)\Big)$.
    \end{corollary}

\subsection{Conditional Submodularity}\label{sec:con-sub-gap}

In \citep{Du2008} and \citep{Zhou2014}, the authors prove approximation results for CDS and Fault-Tolerant CDS respectively, using a special case of submodularity gap. Correspondingly, we generalize the definitions and theorems of subsection 3.1.

  \begin{definition}[$\varepsilon$-approximately $\mathcal{C}$-submodular]\label{def:con-b-gap}
    Let $\mathcal{C}\subseteq 2^U$ and function $f:2^U\rightarrow\mathbb{R}$. The function $f$ is $\varepsilon$-approximately $\mathcal{C}$-submodular if for every $A\subseteq U$, $B\in \mathcal{C}$ and $x \in U\setminus B$ it holds $\Delta_{x}f(A\cup B)\leq \Delta_{x}f(A)+\varepsilon$. 
 \end{definition}
 
\begin{corollary}
    If $\mathcal{C}= 2^U$ then, the above definition is equivalent to unconditional definitions.
\end{corollary}

Using the above definitions, we generalize Lemma~\ref{lem:gap-greedy-approximation}.

\begin{lemma}\label{lem:con-gap-greedy-approximation}
    Let $U$ be a finite set, $\mathcal{C}\subseteq 2^U$ and $f: 2^U\rightarrow\mathbb{R}$ is non-decreasing and $\varepsilon$-approximately $\mathcal{C}$-submodular function.
    Let $S\subseteq U$ with the properties:
    \[
            i) \ S\in D_f, \ \ \ ii) \ f(S)=f_{\max}, \ \ \ iii) \ \delta_{\min}> 0.
        \]
    For every set $C\subseteq U$ with the properties $C\neq\emptyset$, $f(C)=f_{\max}$ and there exists an order of element of $C$ with the property:
    \[
    C=\{c_1, c_2, \cdots, c_{|C|}\}, \ \ C_0=\emptyset, \ \ C_i = \{c_1, c_2, \cdots, c_i\} \ \ and \ \ (\forall i) \ C_i\in \mathcal{C}.\]
    It holds:
    \[
    |S|<\Bigg(1+\frac{\varepsilon}{\delta_{\min}}+\ln\bigg(\frac{\delta_{\max}}{ \delta_{\min}}\bigg)\Bigg)\cdot |C|+1.
    \]
\end{lemma}
\begin{proof}
    We begin the proof with all conditions of Lemma~\ref{lem:gap-greedy-approximation} with the difference the order of the elements of the set $C$ is not arbitrary, but the one determined in the assumptions of this Lemma, that is $C=\{c_1, c_2, \cdots, c_{|C|}\}$.
    By the property of function f, because $f(C)=f_{\max}$ then:
    \[
        f(C\cup S_i) = f_{\max}, \ \forall i=\{1,2,\cdots,|S|\}.
    \]
    So, the next equalities stand:
    \[
    f_{\max}-f(S_{i-1}) = f(S_{i-1}\cup C) - f(S_{i-1}) = \sum_{j=1,\ldots,|C|}\Delta_{c_j}f(S_{i-1}\cup C_{j-1}).
    \]
    The function $f$ is $\varepsilon$-approximately $\mathcal{C}$-submodular and, by the property of $C$, for all $i=\{1,2,\cdots |C|-1\}$, $C_{i}\in \mathcal{C}$. So, the next inequalities stand:
    \[
    \Delta_{c_j}f(S_{i-1}\cup C_{j-1})\leq \Delta_{c_j}f(S_{i-1})+\varepsilon.
    \]
    All the other steps of Lemma~\ref{lem:gap-greedy-approximation} follow accordingly.
\end{proof}

For many problems, such as those with connectivity properties, it may be difficult to define a potential function with a small submodularity gap. For these problems, it may be possible to define a potential function that has a small submodularity gap under some condition $\mathcal{C}$. In Section~\ref{sec:WPPI}, the problem we study has this limitation and we are able to provide a solution using the above analysis.

\begin{theorem}\label{thm:con-beta-greedy}
    Algorithm~\ref{alg:GC} returns a \resizebox{0.8\width}{!}{$\bigg(1+\frac{\varepsilon}{\delta_{\min}}+\ln\Big(\frac{\delta_{\max}}{\delta_{\min}}\Big)\bigg)$}$+O(1)$-approximation to the minimum size set in $\mathcal{Q}$, if and only $f$ is non-decreasing $\varepsilon$-approximately $\mathcal{C}$-submodular function and $f(S) = f_{\max}$, where $S$ is the solution returned by the Algorithm~\ref{alg:GC}, and for the minimum size set $C\in \mathcal{Q}$ there exists an order of element of $C$ with above property:
    \[
    C=\{c_1, c_2, \cdots, c_{|C|}\}, \ \ C_i = \{c_1, c_2, \cdots, c_i\} \ \ and \ \ (\forall i) \ C_i\in \mathcal{C}.\]
\end{theorem}
 \begin{proof}
    Same as Theorem~\ref{thm:beta-greedy}.
\end{proof}

    As in Subsection~\ref{sec:Submodularity gap Maximization}, Corollariess~\ref{cor:simplify-big-o} and \ref{cor:b-to-submodular} are directly implied by Theorem~\ref{thm:con-beta-greedy}.

\section{Fault-tolerant Total Domination}\label{sec:total}

In this section, we define fault-tolerant total domination and provide the first logarithmic approximation for this problem.
As a warm-up, we first consider the case of (standard) total domination.

\subsection{Total Domination}
\label{sec:total-domination}

Recall that given a graph $G = (V, E)$, a subset of nodes $S \subseteq V$ is a total dominating set if 
every node outside of $S$, that is $v \in V\setminus S$, has at least one neighbor in $S$ and $S$ has no isolated nodes.
Equivalently, every node in $V$ has a neighbor in $S$. We recall the following problem:

\begin{definition}[Minimum Total Dominating Set]
    Given a graph $G=(V,E)$ find a total dominating set $S \subseteq V$ of minimum cardinality.
\end{definition}

\noindent For this problem, we define a submodular and monotone increasing function $f$ such that any subset of $V$ achieving $f_{\max}$ is a total dominating set.
    
    \begin{definition}\label{def:f_total}
            Let $G=(V,E)$ be a graph. We define $f: 2^V\rightarrow\mathbb{R}$  as follows
            \[ 
            f(A)=\sum_{v\in V} \delta_A(v),
            \]
            \noindent where
            \[ 
            \delta_A(v)=
            \bigg\{ \begin{array}{l l}
            1, &  |N_A(v)|>0 \\
            0, & otherwise. \\
            \end{array}
            \]
    \end{definition}
    Intuitively, $f(A)$ is the number of nodes having at least one neighbor in $A$.
    
    \begin{definition}
         For $G=(V,E)$ and any $A \subseteq V$, let $K(A) = \{v\in V: \delta_A(v) = 1\}$ be the set of nodes that have at least one neighbor in $A$.
    \end{definition}
    
    \begin{lemma}\label{lem:total_fax}
            Let $G=(V,E)$ be a graph, $A\subseteq V$ and $x\in V\setminus A$. Then:
            \[f(A\cup \{x\})=f(A)+|N_{V\setminus K(A)}(x)|.\]
    \end{lemma}
    \begin{proof} 
        By definition of $K$ and $\delta$, for every $v\in K(A)$ it holds $\delta_{A\cup \{x\}}(v)=\delta_{A}(v)=1$ and for every $v\notin K(A)$ it holds $\delta_{A\cup \{x\}}(v)=\delta_{A}(v)+|N_x(v)|=|N_x(v)|$. We have $f(A\cup\{x\})=
            \sum_{v\in V} \delta_{A\cup\{x\}}(v)$, which we show
        
        \begin{align*}
            \sum_{v\in V} \delta_{A\cup\{x\}}(v)= &\quad
            \sum_{v\in K(A)} \delta_{A\cup\{x\}}(v)&&+\sum_{v\notin K(A)} \delta_{A\cup\{x\}}(v)\\
            =&\quad
            \sum_{v\in K(A)} \delta_{A}(v)&&+\sum_{v\notin K(A)} \big(\delta_{A}(v)+|N_x(v)|\big)\\
            =&\quad
            \sum_{v\in K(A)} \delta_{A}(v)&&+\sum_{v\notin K(A)}\delta_{A}(v)+\sum_{v\notin K(A)}|N_x(v)|
        \end{align*}
        \[
            =\sum_{v\in V} \delta_{A}(v)+\sum_{v\notin K(A)}|N_x(v)|=f(A)+|N_{V\setminus K(A)}(x)|.
        \]
    \end{proof}
    
    \begin{lemma}\label{lem:total-monoton-submodular}
        Function $f$ is submodular and monotone increasing.
    \end{lemma}
    
    \begin{proof}
        First, we show $f$ is monotone increasing.
        For every $A, B$ such that $A\subseteq B\subseteq V$ and for every $v\in V$, by applying definitions, it holds
        $N_A(v)\subseteq N_B(v)$, which implies $\delta_A(v) \leq \delta_B(v)$ and so $f(A)\leq f(B)$.

        Second, we show $f$ is submodular. For every $A, B$ such that $A\subseteq B\subseteq V$ and for every $x\in V\setminus B$, by definition of $K$, it holds:
        \[K(A)\subseteq K(B) \Rightarrow V\setminus K(B)\subseteq V\setminus K(A) \Rightarrow\]
        \[|N_{V\setminus K(B)}(x)| \leq |N_{V\setminus K(A)}(x)|,\]
        which implies  $f(B)+|N_{V\setminus K(B)}(x)|-f(B) \leq f(A)+|N_{V\setminus K(A)}(x)|-f(A)$. By Lemma \ref{lem:total_fax}, it holds:
        \[f(A\cup \{x\})=f(A)+|N_{V\setminus K(A)}(x)| \ \ and \ \ f(B\cup \{x\})=f(B)+|N_{V\setminus K(B)}(x)|\]
        and so $ f(B\cup \{x\})-f(B)\leq f(A\cup \{x\})-f(A)$.
     \end{proof}
    
    \begin{lemma}\label{lem:total-fmax}
     Let $G=(V,E)$ be a graph. A set $S\subseteq V$ is a total dominating set if and only if $f(S)=f_{\max}$.
    \end{lemma}
    
    \begin{proof}
        Note $f_{\max}=f(V)=|V|$ since  for all $v\in V$ it holds $|N_V(v)|> 0$ and so $\delta_V(v)=1$.
        
        Let $S\subseteq V$ be a total dominating set. 
        Since $S$ is total dominating set, then $\delta_S(v)=1$ for all $v\in V$. So, we get $f(S)=\sum_{v\in V} \delta_S(v)=|V|= f_{\max}$.
        
        Consider the case $S\subseteq V$ is not a total dominating set and $f(S)=f_{\max}=|V|$. Then, there exists $w\in V$ such that $N_S(w)=0 \Rightarrow \delta_S(w)=0$. So, $f(S)=\sum_{v\in V} \delta_S(v) =\sum_{v\in V\setminus \{w\}} \delta_S(v) < |V| = f_{\max}$, a contradiction.
    \end{proof}  

    \begin{theorem}\label{thm:total}
        Algorithm \ref{alg:GC}, where $U = V$, returns a $\Big(1+\ln\big(\Delta\big)\Big)$-approximation for Minimum Total Dominating Set.
    \end{theorem}

    \begin{proof}
        By definition of function $f$, Lemmata \ref{lem:total-monoton-submodular} and \ref{lem:total-fmax}, and Corollary~\ref{cor:b-to-submodular}, it holds $S$ is a total dominating set for $G$, $
        \delta_{\max} = \Delta_{s_{1}}f(\emptyset) = |N_V(s_1)|=\Delta$ by greedy choice of $s_1$, $\delta_{\min} = \Delta_{s_{|S|}}f(S_{|S|-1}) \geq 1$ and $
        |S|\leq
        \Bigg(1 + \ln\bigg(\frac{\delta_{\max}}{\delta_{\min}}\bigg)\Bigg)\cdot |C|
        $
        for every total dominating set $C\subseteq V$. 
        Assuming $S^*$ is a total dominating set of minimum size, it follows
        $
        |S|\leq \Big(1+\ln\big(\Delta \big)\Big)\cdot |S^*|
        $.
    \end{proof}
    
\paragraph{Discussion.}
Let us briefly comment on why the analysis performed in  \citep{zhu2009approximation} falls short of our approximation guarantee. 
Below, let $f$ be our potential function (Definition~\ref{def:f_total}), $f'$ be the potential function defined in \citep{zhu2009approximation} and $T^*=\{y_1, \cdots y_{|T^*|}\}$ be the minimum total dominating set for a given input. 
In \citep{zhu2009approximation}, $f'$ comprises two parts $f'(T)=i(T)+w(T)$, where
$i(T)$ denotes the number of nodes in $T$ which are not adjacent to $T$, that is, the number of isolated nodes within $T$, and
$w(T)$ is the number of nodes outside $T$ which are not adjacent to $T$. 
While they prove $w$ is submodular, they do not prove the same for $i$.
To overcome this obstacle, they observe $\Delta_{y_j}i(T_{i-1}\cup T^*_{j-1}) \le \Delta_{y_j}i(T_{i-1})+m_j$, where $m_j=1$ if $y_j$ is not adjacent to $T^*_{j-1}$ and $m_j=0$ otherwise. 
Then, it follows $a_i \leq a_{i-1} - \frac{a_{i-1}}{|T^*|}+\frac{m}{|T^*|}$, where $m=\sum_{j=1}^{|V|}m_j$. 
Since $T^*$ is a total dominating set, they observe $m \leq \frac{|T^*|}{2}$, which leads to $a_i \leq a_{i-1} - \frac{a_{i-1}}{|T^*|}+\frac{1}{2}$.
Instead, in our analysis, we prove $f$ is submodular and as a result we arrive to the inequality $a_i \leq a_{i-1} - \frac{a_{i-1}}{|T^*|}$.
Overall, we improve the result of this technique from $1.5+\ln(\Delta-0.5)$ to $1+\ln(\Delta)$, but the best approximation result remains to $H(\Delta)-0.5 \approx \ln(\Delta)+0.077$ \citep{CHLEBIK20081264}.





\subsection{Fault-tolerant Total Domination}
    In this subsection, we generalize to the fault-tolerant case of total domination.
    Recall that given a graph $G = (V, E)$, a subset of nodes $S \subseteq V$ is a fault-tolerant total dominating set if every node outside of $S$, that is, $v \in V\setminus S$, has at least $m$ neighbors in $S$ and $S$ has no isolated nodes. 
    We recall the following problem:

    \begin{definition}[Minimum Fault-Tolerant Total Dominating Set]
    Given a graph $G=(V,E)$, find a fault-tolerant total dominating set $S \subseteq V$ of minimum cardinality.
    \end{definition}
    
    \noindent For this problem, we define a submodular and monotone increasing function $f$ such that any subset of $V$ achieving $f_{\max}$ is a fault-tolerant total dominating set.
    
    \begin{definition}\label{def:m_A}
            Let $G=(V,E)$ be a graph. 
            We define $f: 2^V\rightarrow\mathbb{R}$ as follows:
            \[ 
            f(A)=\sum_{v\in V} m_A(v),
            \]
            \noindent where:
            \[ 
            m_A(v)=
            \left\{ \begin{array}{lll}
            m, &  \big[v\not\in A \ \wedge |N_A(v)| \geq m\big] \ \vee \ \big[v\in A \ \wedge \ |N_A(v)|>0\big] &  \ (a)\\\\
            m-1, & v\in A \ \wedge \ |N_A(v)|=0 & \ (b)\\\\
            |N_A(v)|, & otherwise. & \ (c) \\
            \end{array}
            \right.
            \]
    \end{definition}
    
    Intuitively,
    $m_A(v)$ is the potential of node $v$ toward satisfying the problem definition. 
    When $v$ fully meets the definition requirements for total domination, it is assigned a value of $m$ (case a). 
    If $v \not\in A$ and has $m' < m$ neighbors in $A$, then we assign it a value of $m'$ (case c).
    If $v\in A$, then it meets the definition when it has at least one neighbor in $A$. 
    In case it does not have a neighbor in $A$, we artificially assign it a value of $m-1$ (case b).
    The value will increase to $m$ only when there appears a neighbor of $v$ in $A$.
    
    \begin{definition}
         For $G=(V,E)$ and $A \subseteq V$, let $K(A) = \{v\in V: m_A(v) = m\}$ be the set of nodes that have at least one neighbor in $A$, if the node is in $A$, or have at least $m$ neighbors in $A$, if the node is in $V\setminus A$.
    \end{definition}
    
    \begin{lemma}\label{lem:ftotal_fax}
            Let $G=(V,E)$, $A\subseteq V$ and $x\in V\setminus A$. Then:
            \[f(A\cup \{x\})=f(A)+|N_{V\setminus K(A)}(x)|+t_A(x),\]
            \noindent where
            \[ 
            t_A(x)=
            \left\{ \begin{array}{ll}
            0, &  x\in K(A)\\
            m-|N_A(x)|, & x\notin K(A) \ \wedge \ |N_A(x)|>0\\
            m-1, & x\notin K(A) \ \wedge \ |N_A(x)|=0 \ . \\
            \end{array}
            \right.
            \]
    \end{lemma}
    
    Before proceeding to the proof, we provide some intuition on the definition of $t_A(x)$.
    Overall, it holds $t_A(x) = m_{A\cup\{x\}}(x) - m_A(x)$, that is, $t_A(x)$ captures the increase in the potential of $x$ attributed to the insertion of $x$ into $A$.

    \begin{proof}
    We first compute an expression for the value of $m_{A\cup \{x\}}(v)$ for any $v \in V$, where $v\neq x$.
    If $v \in K(A)$, then it holds $m_{A\cup \{x\}}(v)=m_{A}(v)=m$. Otherwise, if $v\not\in K(A)$, let us show that
    
    \noindent $m_{A\cup \{x\}}(v)=m_{A}(v)+|N_x(v)|$.
    \begin{itemize}
        \item If $v \in A$, then $|N_A(v)| = 0$ by definition of $K(A)$, so we are in case (b) of Definition~\ref{def:m_A} and  $m_A(v) = m-1$. 
        Since $v \in A$, it also holds $v \in A\cup\{x\}$.
        \begin{itemize}
            \item If $(v,x) \in E$,
            then $|N_{A\cup\{x\}}(v)| > 0$, and by case (a) of Definition~\ref{def:m_A}, it follows
            
            $m_{A\cup\{x\}}(v) = m = m - 1 + 1 = m_{A}(v)+|N_x(v)|$.
            \item If $(v,x) \not\in E$, then $|N_{A\cup\{x\}}(v)| = |N_A(v)| + |N_x(v)| = 0 + 0 = 0$, so
            
            $m_{A\cup\{x\}}(v) = m - 1 = m - 1 + 0 = m_A(v) + |N_x(v)|$.
        \end{itemize}
        
        \item If $v \not\in A$, then $v \not\in A\cup\{x\}$.
        We are in case (c) of Definition~\ref{def:m_A}, so it holds $m_A(v) = |N_A(v)|$. Since $v \not\in A\cup\{x\}$, it holds
        
        $m_{A\cup\{x\}}(v) = |N_{A\cup\{x\}}(v)| = \sum_{u\in A\cup\{x\}}|N_u(v)| = \sum_{u\in A}|N_u(v)| + |N_x(v)| =$
        
        $ |N_A(v)| + |N_x(v)| = m_A(v) + |N_x(v)|$.

        \item  For the new node $x$, it holds $m_{A\cup \{x\}}(x)=m_{A}(x) + t_A(x)$, since 
        \begin{itemize}
            \item If $x \in K(A)$, then $m_{A\cup \{x\}}(x)=m_{A}(x)=m$.
            \item If $x \notin K(A)$ and $|N_A(x)|>0$, then
            
            $ m_{A\cup \{x\}}(x)= m = |N_A(x)| + m - |N_A(x)| = m_{A}(x)+m-|N_A(x)|$.
            \item If $x \notin K(A)$ and $|N_A(x)|=0$, then $m_{A\cup \{x\}}(x)= m-1 = 0 + m-1 = m_{A}(x)+m-1$.
        \end{itemize}
    \end{itemize}

We now compute the value of $f(A\cup\{x\})$:
     \[f(A\cup\{x\}) =
        \sum_{v\in V} m_{A\cup\{x\}}(v) =  
    \]
    \begin{alignat*}{6}
        &\sum_{\substack{\scriptscriptstyle v\in K(A) \\\scriptscriptstyle  v\neq x}} m_{A\cup\{x\}}(v) \ &&+
        &&\sum_{\substack{\scriptscriptstyle v\notin K(A) \\
        \scriptscriptstyle  v\neq x}} m_{A\cup\{x\}}(v) \ &&+
        &&m_{A\cup \{x\}}(x) = \nonumber
        &&\\
        &\sum_{\substack{\scriptscriptstyle v\in K(A) \\\scriptscriptstyle  v\neq x}} m_{A}(v) \ &&+
        &&\sum_{\substack{\scriptscriptstyle v\notin K(A) \\\scriptscriptstyle  v\neq x}} \big(m_{A}(v)+|N_x(v)|\big) \ &&+
        &&m_{A}(x) + m_{A\cup \{x\}}(x)-m_{A}(x) = \nonumber
        &&\\
        &\sum_{\substack{\scriptscriptstyle v\in K(A) \\\scriptscriptstyle  v\neq x}} m_{A}(v) \ &&+ 
        &&\sum_{\substack{\scriptscriptstyle v\notin K(A) \\\scriptscriptstyle  v\neq x}}m_{A}(v) \ \ \ \ + m_{A}(x)&&+ 
        && \sum_{\substack{\scriptscriptstyle v\notin K(A) \\\scriptscriptstyle  v\neq x}}|N_x(v)| + t_A(x) = \nonumber
        &&
    \end{alignat*}
    
    \[ \sum_{v\in V} m_{A}(v) + 
        \sum_{v\notin K(A)}|N_x(v)| +
        t_A(x)=  
        f(A)+|N_{V\setminus K(A)}(x)| + t_A(x).
        \nonumber 
    \]
    \end{proof}
    
    \begin{lemma} \label{lem:F-monoton-submodular}
        Function $f$ is submodular and monotone increasing.
    \end{lemma}

     \begin{proof}
        First, we show $f$ is monotone increasing. 
        We prove $m_A(v)\leq m_B(v)$ $\forall A\subseteq B$ and $\forall v\in V$. 
        We distinguish the cases below:
        \begin{itemize}
            \item Let $v\in A$ and $|N_A(v)|>0$. 
            Then, $v\in B$ and $|N_B(v)|>0$. So, $m_A(v)= m_B(v) = m$.
            
            \item Let $v\in A$ and $|N_A(v)|=0$. 
            Then, $v\in B$.
            \begin{itemize}
                \item If $|N_B(v)|=0$, then $m_A(v)= m_B(v) = m-1$.
                \item If $|N_B(v)|>0$, then $m-1=m_A(v)< m_B(v)=m$.
            \end{itemize}
            
            \item Let $v\notin A$, $|N_A(v)|>0$ and $v\in B$. Then, $|N_A(v)|=m_A(v)\leq m_B(v)=m$.
            
            \item Let $v\notin A$, $|N_A(v)|=0$ and $v\in B$. Then, $0=m_A(v)\leq m_B(v)$.
            
            \item Let $v\notin A$ and $v\notin B$. Then, $|N_A(v)|=m_A(v)\leq m_B(v)=|N_B(v)|$.
        \end{itemize}
        Second, we show $f$ is submodular. 
        For every $A, B$ such that $A\subseteq B\subseteq V$ and for every $x\in V\setminus B$, by definition of $K$, it holds $K(A)\subseteq K(B)$, which implies $V\setminus K(B)\subseteq V \setminus K(A)$ and so
        \[|N_{V\setminus K(B)}(x)| \leq |N_{V\setminus K(A)}(x)|.\]
        
        \noindent It then follows $f(B)+|N_{V\setminus K(B)}(x)|-f(B) \leq f(A)+|N_{V\setminus K(A)}(x)|-f(A)$. 
        To complete the proof, it suffices to show $t_B(x) \leq t_A(x)$, since in that case it follows
        \[f(B)+|N_{V\setminus K(B)}(x)|+t_B(x)-f(B) \leq f(A)+|N_{V\setminus K(A)}(x)|+t_A(x)-f(A)\]
        and by Lemma \ref{lem:ftotal_fax} 
        we have $f(B\cup \{x\})-f(B)\leq f(A\cup \{x\})-f(A)$.
        
        In order to prove $t_B(x) \leq  t_A(x)$, we distinguish the cases below:
        \begin{itemize}
            \item If $x\in K(A)$, then $x\in K(B)$, and so    $t_A(x)=t_B(x)=0$.
             
             \item If $x \notin K(A)$ and $|N_A(x)|=0$, then
             
             \begin{itemize}
                 \item If $x \in K(B)$, then $0=t_B(x) \leq t_A(x)=m-1$.
                 \item If $x \notin K(B)$ and $|N_B(x)|=0$, then $t_A(x)=t_B(x)=m-1$.
                 \item If $x \notin K(B)$ and $|N_B(x)|>0$, then $t_B(x) = m-|N_B(x)|\leq m-1 = t_A(x)$.
             \end{itemize}
             
             \item If $x \notin K(A)$ and $|N_A(x)|>0$, then it follows $|N_A(x)|<m$ and $|N_B(x)|>0$.
              \begin{itemize}
                 \item If $x \in K(B)$, then $0=t_B(x) < t_A(x)=m-|N_A(x)|$.
                 \item If $x \notin K(B)$, then: 
                 \[|N_A(x)|\leq |N_B(x)|  \Rightarrow m-|N_B(x)|\leq m-|N_A(x)| \Rightarrow t_B(x) \le t_A(x).\]
             \end{itemize}
        \end{itemize}
    \end{proof}

    \begin{lemma}\label{lem:ftotal-fmax}
     Let $G=(V,E)$ be a graph. A set $S\subseteq V$ is a fault-tolerant total dominating set if and only if $f(S)=f_{\max}$.
    \end{lemma}
    
    \begin{proof}
        We get $f_{\max}=f(V)=m|V|$, since, for all $v \in V$, it holds $|N_V(v)| > 0$ and so $m_V(v)=m$.
        
        Let $S\subseteq V$ be a fault-tolerant total dominating set.
        Then, $m_V(v)=m$ for all $v\in V$. So,
        
        \[f(S)=\sum_{v\in V} m_S(v)=m|V|= f_{\max}.\]
        
        Assume $S\subseteq V$ is not a fault-tolerant total dominating set and $f(S)=f_{\max}=m|V|$. Then, there exists $v'\in V\setminus S$ such that $|N_S(v')|<m$ or there exists $s'\in S$ such that $|N_S(s')|=0$. So, there exists $w\in V$ such that $m_S(w)<m$ since $|N_S(v')|<m \Rightarrow m_S(v')<m $ and $|N_S(s')|=0 \Rightarrow m_S(s')=m-1<m$. Thus, $f(S)=\sum_{v\in V\setminus\{w\}}m_S(v) +  m_S(w) < m(|V|-1)+m=f_{\max}$, a contradiction.
    \end{proof} 
    
    \begin{theorem}\label{thm:fault-tolerant-total}
        Algorithm \ref{alg:GC}, where $U = V$, returns a $\big(1+\ln\big(\Delta+m-1\big)\big)$-approximation for Minimum Fault-Tolerant Total Dominating Set.
    \end{theorem}

    \begin{proof}
        By definition of $f$, Lemmata
        \ref{lem:F-monoton-submodular} and \ref{lem:ftotal-fmax}, and Corollary~\ref{cor:b-to-submodular}, it holds $S$ is a fault-tolerant total dominating set for $G$ and $\delta_{\max} = \Delta_{s_{1}}f(\{\emptyset\}) = |N_V(s_1)|+t_\emptyset(s_1) = \Delta + m - 1$, $|N_V(s_1)|=\Delta$ by greedy choice of $s_1$ and $t_\emptyset(s_1) = m - 1$ by case b of Definition~\ref{def:m_A}.
        It also holds: $$\delta_{\min} = \Delta_{s_{|S|}}f(S_{|S|-1}) \geq 1 \ and \
        |S|\leq
        \Big(1 + \ln\big(\frac{\delta_{\max}}{\delta_{\min}}\big)\Big)\cdot |C|
        $$
        for every fault-tolerant total dominating set $C$. 
        Assuming $S^*$ is a fault-tolerant total dominating set of minimum size, it follows
        $
        |S|\leq \Big(1+\ln\big(\Delta+m-1 \big)\Big)\cdot |S^*|
        $
    \end{proof}

\section{Weighted Partial Positive Influence Domination}\label{sec:WPPI}
    In this section, we present a first approximation algorithm for the Weighted Partial Positive Influence Connected Dominating Set problem. To achieve the approximation, we apply the framework in subsection \ref{sec:con-sub-gap}.
    As a warm-up, we first consider the unconnected case of the problem.
    
    \subsection{Weighted Partial Positive Influence Dominating Set}\label{sec:WPPIDS}
    Recall that given a graph $G = (V, E, w)$, a subset of nodes $S \subseteq V$ is a Weighted Partial Positive Influence Dominating Set if every node outside of $S$, that is $v\in V\setminus S$, it holds $W_S(v)\geq W(v)/2$.
    We define the following problem:

    \begin{definition}[Minimum Weighted Partial Positive Influence Dominating Set (WPPIDS)]
    Given a $\ \ $ graph $G=(V, E, w)$, find a WPPIDS $S \subseteq V$ of minimum cardinality.
    \end{definition}
    
    \noindent For this problem,
    we define a submodular and monotone increasing function $h$ such that any subset of $V$ achieving $h_{\max}$ is WPPIDS.

 \begin{definition}\label{def:wm_A}
            Let $G=(V,E,w)$ be a graph. 
            We define $h: 2^V\rightarrow\mathbb{R}$ as follows:
            \[ 
            h(A)=\sum_{v\in V} h_A(v),
            \]
            \noindent where:
            \[ 
            h_A(v)=
            \left\{ \begin{array}{ll}
            W(v)/2, \, & v\in A \ \vee \ W_A(v)\geq W(v)/2\\\\
            W_A(v), \, & otherwise.\\
            \end{array}
            \right.
            \]
    \end{definition}

    Before we continue, we define some important concepts that we use in the theorems. At first, we define function $ms:V\rightarrow\mathbb{R}$ as the minimum positive gain that every node $v\in V$ gives:
   
    \[
    ms(v) = \min\limits_{\substack{A\subseteq V \\ W(v)/2 - W_A(v)>0 }}\Big(W(v)/2 - W_A(v)\Big).
    \]

    \noindent Let $v\in V$, $w_0=W(v)/2$ and $w_1, \dots, w_d$ be the weights for every edge between the node $v$ and a neighbor (where d is the degree of $v$). We rewrite every weight as a reduced fraction ($p_i,q_i\in\mathbb{N}$):

    \[
    w_0 = \frac{p_0}{q_0}, w_1 = \frac{p_1}{q_1}, \dots, w_d = \frac{p_d}{q_d} \ ,
    \]

    \noindent we define function $l:V\rightarrow\mathbb{R}$ as the least common multiple of denominators of reduced fractions of weights of every node:

    \[
    l(v) = LCM\{q_0,q_1, \dots,q_d\}
    \]

    \noindent and we set:
    \[
        L = \max\limits_{v\in V} l(v)
    \]
    We have the next Lemma:
    \begin{lemma}\label{lem:h_inc_sub}
        \[
        \frac{1}{ms(v)}\leq l(v)\leq L, \forall v\in V.
        \]
    \end{lemma}
    \begin{proof}
    Let $v\in V$. By using the above analysis, we rewrite every $w_i$ as follows ($p'_i\in \mathbb{N}$):
    \[
    w_0 = \frac{p'_0}{l(v)}, w_1 = \frac{p'_1}{l(v)}, \dots, w_d = \frac{p'_d}{l(v)} \ .
    \]
    Then:
    \[
    ms(v) = w_0 - (w_{1}+\cdots+w_{d}) = 
    \frac{p'_0}{l(v)} - \Big(\frac{p'_{1}}{l(v)} + \cdots + \frac{p'_{d}}{l(v)}\Big) \Rightarrow
    \]
    \[
    ms(v) = 
    \frac{p'_0 - (p'_{1}+\cdots+p'_{d})}{l(v)} > 0 \Rightarrow ms(v) \geq \frac{1}{l(v)}.
    \]
    \end{proof}
    
    \begin{lemma}\label{lem:WPPIDS-monoton-submodular}
        Function $h$ is submodular and monotone increasing.
    \end{lemma}
    \begin{proof}
        By definition of $W_A(v)$, $h$ is monotone increasing. We show $h$ is submodular. If we show that $\Delta_{x}h_B(v)\leq \Delta_{x}h_A(v)$ $\big($where $\Delta_{x}h_B(v)=h_{B\cup \{x\}}(v)-h_B(v)\big)$ for every $A\subseteq B\subseteq V$ and $x\in V$, then, it holds:
        \begin{align*}     
           \Delta_{x}h_B(v) &\leq 
            \Delta_{x}h_A(v), \ \forall v\in V \then \\
            h(B)+\sum_{v\in V}\Delta_{x}h_B(v)-h(B) &\leq 
            h(A)+\sum_{v\in V}\Delta_{x}h_A(v)-h(A) \then \\
            h(B\cup \{x\})-h(B) &\leq h(A\cup \{x\})-h(A)
        \end{align*}
    because for every $A\subseteq V$ and $x\in V$:
    \[
    h(A\cup \{x\}) = \sum_{v\in V} h_{A\cup \{x\}}(v) = 
    \sum_{v\in V} h_{A}(v)+\sum_{v\in V} h_{A\cup \{x\}}(v) - \sum_{v\in V} h_{A}(v) =
    \]
    \[
    h(A)+\sum_{v\in V} \big(h_{A\cup \{x\}}-h_{A}(v)\big) =
     h(A)+\sum_{v\in V}\Delta_{x}h_A(v).
    \]
     By definition of $W_A(v)$ and $h_A(v)$, $\forall v\notin A\cup\{x\}$, it holds:
      \[ 
            \Delta_{x}h_A(v)=
            \left\{ \begin{array}{lll}
              W_x(v), \, & W_A(v)+W_x(v)\leq W(v)/2\\\\
            W(v)/2-W_A(v), \, & W_A(v)\leq W(v)/2 \leq W_A(v)+W_x(v)\\\\
            0, \, & otherwise.\\
            \end{array}
            \right.
            \]
    and $\forall v\in A\cup\{x\}$ it holds $\Delta_{x}h_A(v)=0$. Let $A\subseteq B\subseteq U$ and $x\in U$. We distinguish the cases below:
        \begin{itemize}
            \item Let $v\in A\cup\{x\}$. Then, $v\in B\cup\{x\}$ and so $\Delta_{x}h_A(v)=\Delta_{x}h_B(v)=0$.
            \item Let $v\in B\cup\{x\}$ and $v\notin A\cup\{x\}$. Then $\Delta_{x}h_B(v)=0\leq \Delta_{x}h_A(v)$.
            \item Let $v\notin B\cup\{x\}$ and $v\notin A\cup\{x\}$. Then, because $W_S(v)$ is increasing, then $\Delta_{x}h_B(v)\leq\Delta_{x}h_A(v)$.
        \end{itemize}
    \end{proof}

    \begin{lemma}\label{lem:WPPIDS-fmax}
         Let $G=(V,E,w)$ be a graph. A set $S\subseteq V$ is a WPPIDS if and only if $h(S)=h_{\max}$.
    \end{lemma}
    \begin{proof}
        We get $f_{\max}=f(V)=\sum_{v\in V} W(v)/2$, since, for all $v \in V$, it holds $W_V(v) = W(v)\geq W(v)/2$ and so $h_V(v)=W(v)/2$.
        
        Let $S\subseteq V$ be a WPPIDS.
        Then, $h_S(v)=W(v)/2$ for all $v\in V$. We have $f(S)=\sum_{v\in V} W(v)/2 = f(V) = f_{\max}$.
        
        Assume $S\subseteq V$ is not a WPPIDS and $h(S)=h_{\max}$. Then, there exists $v'\in V\setminus S$ such that $W_A(v)<W(v)/2$. So,
        \begin{align*}
            f(S) = &\sum_{v\in V} h_S(v) = \sum_{v\in V\setminus \{v'\}} h_S(v) + h_S(v') \leq
            \sum_{v\in V\setminus \{v'\}} W(v)/2 + h_S(v') = \\
            &\sum_{v\in V\setminus \{v'\}} W(v)/2 + W_S(v') < \sum_{v\in V\setminus \{v'\}} W(v)/2 + W(v')/2 = \sum_{v\in V} W(v)/2 = f_{\max}.
        \end{align*}
        We have a contradiction.
    \end{proof} 

    \begin{theorem}\label{thm:WPPIDS}
        Algorithm \ref{alg:GC}, where $U = V$, returns a \resizebox{0.8\width}{!}{$\Big(1+\ln\big(\delta_{\max}/\delta_{\min}\big)\Big)$}
        -approximation for WPPIDS, where:
        \[
        \delta_{\max} = \frac{3}{2}\cdot W(s_1) \ \ and  \ \ \delta_{\min} \geq \frac{1}{L}.
        \]
    \end{theorem}
    \begin{proof}
        By Lemma~\ref{lem:h_inc_sub}, it holds $\delta_{\min} \geq 1/L$. So, the results hold by definition of $h$, Lemmata~\ref{lem:WPPIDS-monoton-submodular} and \ref{lem:WPPIDS-fmax}, and Corollary~\ref{cor:b-to-submodular} (same approach with Theorem~\ref{thm:fault-tolerant-total}).
    \end{proof}

    \paragraph{Discussion.}

    This approximation result is a generalization of the result which is presented in \citep{zhu2010new}, where they prove a $1+\ln(\lceil\frac{3\Delta}{2}\rceil)$-approximation for the case that all weights be equal to 1. If we let $\lceil W(u)/2\rceil$ be an upper bound for function $h_A(u)$, then it follows $L=1\then \delta_{\min}=1$, $W(s_1)=\Delta\then\delta_{\max}=\lceil \frac{3\Delta}{2}\rceil$ and we obtain the same result. Also, the approximation results can be extended to degree percentage constraints problem \citep{zhong2023unified} (simple case) to provide a $\Big(1+\ln\big((1 + p)\cdot L\cdot W\big)\Big)$-approximation guarantee (see discussion, subsection~\ref{sec:WPPICDS}).


\subsection{Weighted Partial Positive Influence Total Dominating Set}\label{sec:WPPITDS}
    Recall that given a graph $G = (V, E, w)$, a subset of nodes $S \subseteq V$ is a Weighted Partial Positive Influence Total Dominating Set if every node outside of $S$, that is $v\in V\setminus S$, it holds $W_S(v)\geq W(v)/2$ and $S$ has no isolated nodes.
    We define the following problem:

    \begin{definition}[Minimum$\ $ Weighted$\ $ Partial$\ $ Positive$\ $ Influence$\ $ Total$\ $ Dominating$\ $ Set$\ $(WPPIDS)]
    Given a graph $G=(V, E, w)$, find a WPPITDS $S \subseteq V$ of minimum cardinality.
    \end{definition}

    \noindent For this problem,
    by using the functions we define at subsection~\ref{sec:total-domination} (function $f$) for Total Domination and at subsection~\ref{sec:WPPIDS} (function $h$) for WPPIDS, we can prove an approximation result for WPPITDS problem. 

     \begin{definition}
            Let $G=(V,E,w)$ be a graph. 
            We define $g: 2^V\rightarrow\mathbb{R}$ as follows:
            \[ 
            g(A)=h(A)+\frac{1}{L}f(A),
            \]
            \noindent where $f$ is the function that is defined in subsection~\ref{sec:total-domination} and $h$ is the function that is defined in subsection~\ref{sec:WPPIDS}.
    \end{definition}

    \begin{lemma}\label{lem:WPPITDS-monoton-submodular}
        Function $g$ is submodular and monotone increasing.
    \end{lemma}
    \begin{proof}
        The results holds by definition of g and Lemmata \ref{lem:total-monoton-submodular} and \ref{lem:WPPIDS-monoton-submodular}.
    \end{proof}

    \begin{lemma}\label{lem:WPPITDS-fmax}
         Let $G=(V,E,w)$ be a graph. A set $S\subseteq V$ is a WPPTIDS if and only if $g(S)=g_{\max}$.
    \end{lemma}
    \begin{proof}
        The results holds by definition of g and Lemmata \ref{lem:total-fmax} and \ref{lem:WPPIDS-fmax}.
    \end{proof}

    \begin{theorem}\label{thm:WPPITDS}
        Algorithm \ref{alg:GC}, where $U = V$, returns a \resizebox{0.8\width}{!}{$\Big(1+\ln\big(\delta_{\max}/\delta_{\min}\big)\Big)$}-approximation for WPPITDS $\Big(= 1+\ln\big(\frac{3}{2}\cdot L\cdot W+\Delta\big)\Big)$, where:
        \[
        \delta_{\max} = \frac{3}{2}\cdot W(s_1) + \frac{\Delta}{L} \ \ \ and  \ \ \ \delta_{\min} = \frac{1}{L}.
        \]
    \end{theorem}
    \begin{proof}
        The result holds by definition of $g$, Lemmata~\ref{lem:WPPITDS-monoton-submodular} and \ref{lem:WPPITDS-fmax}, and Corollary~\ref{cor:b-to-submodular} (same approach with Theorem~\ref{thm:fault-tolerant-total}).
    \end{proof}

\subsection{Weighted Partial Positive Influence Connected Dominating Set}\label{sec:WPPICDS}
    Recall that given a graph $G = (V, E, W)$, a subset of nodes $S \subseteq V$ is a Weighted Partial Positive Influence Connected Dominating Set if the subset is WPPIDS and connected.
    We define the following problem:
    
    \begin{definition} \normalfont{\textbf{(Minimum$\ $ Weighted$\ $ Partial$\ $ Positive$\ $ Influence$\ $ Connected$\ $ Dominating$\ $ Set (WPPICDS))}}
    Given a graph $G=(V, E, w)$, find a WPPCIDS $S \subseteq V$ of minimum cardinality.
    \end{definition}

    \noindent For this problem, 
    we define a $\varepsilon$-approximately $\mathcal{C}$-submodular and monotone increasing function $f$ such that any subset of $V$ achieving $f_{\max}$ is WPPICDS.

\begin{definition}
            Let $G=(V,E,W)$ be a graph. 
            We define $f: 2^V\rightarrow\mathbb{R}$, $\forall A\subseteq V$, as follows:
            \[ 
            f(A)=h(A)+c(A), \ \ c(A)=\frac{1}{L}(|V|-q(A)-p(A)),
            \]
            \noindent where $p(A)$ is the number of components of induced subgraph $G_A$ and $q(A)$ the number of components of spanning subgraph induced by the edge set $\{e\in E: e \ \text{has at least one end in } A\}$.
    \end{definition}
    
     \begin{lemma}\label{lem:WPPICDS-subgap}
        Function $f$ is $\varepsilon$-approximately \resizebox{0.9\width}{!}{$\mathcal{C}$}-submodular and non-decreasing, where $\varepsilon=1/L$ and $\mathcal{C}$ the collection of all connected subsets of $V$.
    \end{lemma}
    \begin{proof}
        First, we show that $f$ is non-decreasing. By Lemma~\ref{lem:WPPIDS-monoton-submodular}, $h(A)$ is monotone increasing. By \cite{zhu2010new}, function $c(A)$ is non-decreasing. So, the result holds.
        
        Second, we show that $f$ is $\varepsilon$-approximately $\mathcal{C}$-submodular, where $\varepsilon=1/L$.
        By Lemma~\ref{lem:WPPIDS-monoton-submodular}, it holds:
        \[
        \Delta_{x}h(B)\leq \Delta_{x}h(A), \forall x\in V \ and \ \forall A\subseteq B.
        \]
        By \cite{Ruan2004} (Lemma 3.2), it holds:
        \[
        -\Delta_{x}q(B)\leq -\Delta_{x}q(A),\forall x\in V \ and \ \forall A\subseteq B.
        \]
        By \cite{Zhou2014} (Lemma 4.1, claim 3), it holds:
        \[
        -\Delta_{x}p(B)\leq -\Delta_{x}p(A)+1, \forall x\in V \ and \ \ (\forall A\subseteq B) \ B\setminus A \ connected.
        \]
        By above inequalities, it holds:
        \[
        \Delta_{x}f(B)\leq \Delta_{x}f(A)+\frac{1}{L}.
        \]
    \end{proof}

    \begin{lemma}\label{lem:WPPICDS-iff}
         Let $G=(V,E,w)$ be a graph. A set $S\subseteq V$ is a WPPICDS if and only if $f(S)=h_{\max}+\frac{|V|-2}{L}$.
    \end{lemma}
    \begin{proof}
    By Lemma~\ref{lem:WPPIDS-monoton-submodular}, $h(S)=h_{\max}$ if and only if $S$ is WPPIDS. By \cite{zhu2010new}, the function $|V|-q(S)-p(S)$ is non-decreasing and by \cite{Ruan2004} $|V|-q(S)-p(S)=|V|-2$ if and only if $S$ is CDS. Also, by definition, $q(A)\geq 1$ and $p(A)\geq 1, \forall A\in V$ and so $|V|-q(A)-p(A)\leq|V|-2$. We have the two cases:
    \begin{itemize}
    \setlength\itemsep{0.5em}
        \item if $S$ is not WPPIDS, then $f(S)<h_{\max}+(|V|-2)/L$.
        \item if $S$ is WPPIDS, then $h(S)=h_{\max}$ and $S$ is DS. So:
        \vspace{0.5em}
        \begin{itemize}
        \setlength\itemsep{0.5em}
            \item if $S$ is not connected then $|V|-q(S)-p(S)<|V|-2$.
            \item if $S$ is connected then $|V|-q(S)-p(S)=|V|-2$.
        \end{itemize}
    \end{itemize}
    So, the result holds.
    \end{proof}

     \begin{theorem}
      Algorithm \ref{alg:GC}, where $U = V$, returns a \resizebox{0.9\width}{!}{$\Big(2+\ln\big(\delta_{\max}/\delta_{\min}\big)\Big)$}-approximation for WPPICDS $\Big(= 2+\ln\big(\frac{3}{2}\cdot L\cdot W + \Delta \big)\Big)$, where:
        \[
        \delta_{\max} = \frac{3}{2}\cdot W(s_1) + \frac{\Delta}{L}
        \ \ and \ \
        \delta_{\min} = \frac{1}{L}.
        \]
    \end{theorem}
    \begin{proof}
        Let $S\in V$ be a solution returned by Algorithm~\ref{alg:GC}. If we prove that $f(S)=f_{\max}$, then the result holds by definition of $f$, Lemmata
        \ref{lem:WPPICDS-subgap}, \ref{lem:WPPICDS-iff} and \ref{lem:con-gap-greedy-approximation}, and Corollary~\ref{cor:simplify-big-o} (same approach with Theorem~\ref{thm:fault-tolerant-total}).
        
        By Algorithm~\ref{alg:GC}, because $h(A)$ is monotone increasing and $c(A)$ is not decreasing, $h(S)=h_{\max}$. So, if we prove that $S$ is connected then, by Lemma~\ref{lem:WPPICDS-iff}, $f(S)=f_{\max}$. 

        Let $S$ is not connected. Because $h(S)=h_{\max}$, then $S$ is dominating set. By using the analysis in \citep{Zhou2014} (Lemma 4.2, Claim 3), there exists $v\in V\setminus S$ such that $c(S\cup \{v\})-c(S)>0$ and so $\Delta_{v}f(S)>0$. We have a contradiction because, by Algorithm~\ref{alg:GC}, $\Delta_{x}f(S)=0, \ \forall x\in V$.
    \end{proof}

    \paragraph{Discussion.}
    In the construction of the potential function for WPPICDS, we introduce a scaling factor of $1/L$ to the cost function $c(A)$. This normalization is essential due to the non-submodular and integer-valued nature of $c(A)$. Without this adjustment, the function exhibits a submodularity gap of $\varepsilon = 1$, which propagates into the approximation guarantee as a dependency on $1/\delta_{\min}$. Since $\delta_{\min}$ cannot be computed exactly for the $h$ function (see subsection~\ref{sec:WPPIDS}) and only a lower bound is available (specifically, $\delta_{\min} \geq 1/L$), the bound would instead involve a factor proportional to $L$. As $L$ can be arbitrarily large depending on the instance, the resulting approximation ratio would be practically uninformative without this normalization.
    This approximation result is a generalization of the result presented in \citep{zhu2010new}, where they prove a $2+\ln(\lceil\frac{5\Delta}{2}\rceil)$-approximation when all weights are equal to $1$. If we let $\lceil W(u)/2\rceil$ be an upper bound for function $h_A(u)$, then it follows $\delta_{\min}=1$, $\delta_{\max}=\lceil 3W(s_1)/2\rceil$, $W(s_1)=\Delta$ and we obtain the same result.  Also, the approximation results can be extended to degree percentage constraints problem \citep{zhong2023unified} (connected case) by substituting the term $1/2$ to a percentage $p$ in function $h_A$ as follows:
     \[ 
            h_A(v)=
            \left\{ \begin{array}{ll}
            p\cdot W(v), \, & v\in A \ \vee \ W_A(v)\geq p\cdot W(v)\\\\
            W_A(v), \, & otherwise.\\
            \end{array}
            \right.
        \]
     to provide a $\Big(2+\ln\big((1 + p)\cdot L\cdot W + \Delta \big)\Big)$ approximation guarantee, with the only difference being that the factor $\frac{3}{2} \cdot W(s_1)$ changes to $(1 + p) \cdot W(s_1)$ and for the calculation of $L$ we set  $w_0=p\cdot W(v)$, $\forall u\in V$.

\section{Conclusions}\label{sec:conclusions}
We extended the general approximation framework for non-submodular functions from integer-valued to fractional-valued functions.
We prove a first logarithmic approximation for Fault-Tolerant Total Domination.
Also, we prove a first logarithmic approximation for Partial Positive Influence Domination problems with fraction-weighted edges in the simple, total, and connected case. Furthermore, all of the above approximation results hold even under degree percentage constraints. In the future, we plan to apply the framework to problems involving required and forbidden properties, e.g., in biology inspired applications like disease pathway detection \citep{nacher2016minimum} or others \citep{grady2022domination}.

\bibliographystyle{abbrvnat}
\bibliography{references}
\label{sec:biblio}

\end{document}